\documentclass[conference]{IEEEtran}
\IEEEoverridecommandlockouts
\usepackage[T1]{fontenc}
\usepackage{cite}
\usepackage{amsmath,amssymb,amsfonts}
\usepackage{mathtools}
\usepackage{graphicx}
\usepackage{textcomp}
\usepackage{xcolor}
\usepackage{algorithm}
\usepackage{algorithmic}
\usepackage{latexsym}
\usepackage{ascmac}
\usepackage[caption=false,font=footnotesize]{subfig} % preamble に追加
\def\BibTeX{{\rm B\kern-.05em{\sc i\kern-.025em b}\kern-.08em
    T\kern-.1667em\lower.7ex\hbox{E}\kern-.125emX}}
\usepackage{tikz}
\usetikzlibrary{arrows.meta,positioning,fit,calc}

\newcommand{\ut}{{\cal UT}}
\newcommand{\luttp}{{\cal LU\!T\!T\!P}}

\newcommand{\mybk}{\hspace{-0.11em}}

\newcommand{\vs}{V{\mybk}S}
\newcommand{\ins}{I{\mybk}n{\mybk}S}
\newcommand{\h}{{\cal H}}

\usepackage{amsthm}

\theoremstyle{plain}
\newtheorem{theorem}{Theorem}
\newtheorem{lemma}[theorem]{Lemma}

\theoremstyle{definition}
\newtheorem{definition}[theorem]{Definition}

\theoremstyle{remark}

\begin{document}

\title{
Efficient Pattern Matching for Unordered Term Tree Patterns under Generalized Height-Constrained Bindings
\thanks{This is an author preprint of a paper presented at the 22nd International Conference on E-Service and Knowledge Management (ESKM 2026), IIAI-AAI 2026. The paper was accepted for inclusion in the conference proceedings published by IEEE Computer Society Conference Publishing Services (CPS).}
\thanks{This work was partially supported by JSPS KAKENHI Grant Numbers JP24K15074 and JP24K15090.}
}

\author{
    \IEEEauthorblockN{
        Shintaro Matsushita\IEEEauthorrefmark{1}\IEEEauthorrefmark{3},
        Takayoshi Shoudai\IEEEauthorrefmark{1}\IEEEauthorrefmark{4}, and
        Yusuke Suzuki\IEEEauthorrefmark{2}
    }
    \IEEEauthorblockA{
        \IEEEauthorrefmark{1}Department of Computer Science and Engineering, Fukuoka Institute of Technology, Fukuoka 811-0295, Japan\\
        Email: \IEEEauthorrefmark{3}mfm25113@bene.fit.ac.jp, \IEEEauthorrefmark{4}shodai@fit.ac.jp
    }
    \IEEEauthorblockA{
        \IEEEauthorrefmark{2}Faculty of Information Sciences, Hiroshima City University, Hiroshima 731-3194, Japan\\
        Email: y-suzuki@hiroshima-cu.ac.jp
    }
}

\maketitle

\begin{abstract}
Unordered trees are useful for modeling hierarchical structures in which the order among siblings is irrelevant. To represent flexible structural patterns in such data, unordered term tree patterns with height-constrained variables provide a natural framework. In our previous work, we studied the pattern matching problem for rooted unordered term tree patterns with height-constrained variables under the restriction that the child port of each variable must correspond to a leaf of a binding tree. In this paper, we remove this restriction and generalize the binding model so that the child port may correspond to any non-root vertex of a binding tree. Under generalized bindings, we formulate the corresponding membership problem and present a polynomial-time pattern matching algorithm. We also implement the proposed algorithm and conduct computational experiments to evaluate its running time. The experimental results show that the proposed method achieves practical running times.
\end{abstract}


%%%%%%%%%%%%%%%%%%%%%%%%%%%%%%%%%%%%%%%%%%%%%%%%%%%%%%%%%%%%%%%%%%%%%%
%% 1. Introduction
%%%%%%%%%%%%%%%%%%%%%%%%%%%%%%%%%%%%%%%%%%%%%%%%%%%%%%%%%%%%%%%%%%%%%%

\section{Introduction}\label{sec:intro}

Unordered trees arise in a variety of application domains in which the relative order among siblings is irrelevant, such as abstract syntax representations, chemical structures, and hierarchical data with unordered branching. Tree patterns on such structures provide a natural way to describe flexible structural constraints. In particular, unordered term tree patterns with height-constrained variables form a useful framework, since each variable can represent a tree-structured component subject to constraints on its trunk length and height. In our previous work~\cite{matsushita-eskm2025}, we studied the pattern matching problem for rooted unordered term tree patterns with height-constrained variables and presented a polynomial-time algorithm for deciding whether a given unordered tree belongs to the language generated by a given pattern. In that model, however, the child port of each height-constrained variable was required to correspond to a leaf of a binding tree.

For rooted unordered trees, this leaf restriction is not essential, because no sibling order needs to be preserved in a binding operation. Motivated by this observation, in this paper we generalize the notion of bindings for height-constrained variables so that the child port may correspond to any non-root vertex of a binding tree. For example, in abstract syntax representations and chemical structure trees, a constrained variable may naturally attach at an internal vertex of a bound subtree rather than only at a leaf. In such cases, the previous leaf restriction excludes patterns that are structurally meaningful in practice. Under generalized bindings, we formulate the corresponding membership problem, develop a polynomial-time pattern matching algorithm, and report computational results obtained from an implementation of the proposed algorithm.

For unordered trees, tree inclusion and tree pattern matching have been studied as fundamental problems on structured data. Classical and more recent results on unordered tree inclusion include those of Kilpel\"ainen and Mannila~\cite{kilpelainen-mannila1995} and Akutsu et al.~\cite{akutsu-tcs2021}. In our research framework, studies on term tree patterns have progressed from unordered tree patterns with internal variables~\cite{shoudai-fct2001} to efficient matching for ordered term tree patterns~\cite{suzuki-ieice2015}. The framework of height-constrained variables was then studied in the ordered setting from a language-theoretic and inductive-inference viewpoint~\cite{shoudai-ieice2017}, and later extended to unordered term tree patterns of bounded dimension~\cite{shoudai-ieice2018}. Our previous work~\cite{matsushita-eskm2025} considered the rooted unordered case under a leaf-restricted binding model.

Our main result is that the above membership problem under generalized bindings is solvable in polynomial time. More precisely, we present a dynamic-programming-based matching algorithm, analyze its time complexity, and evaluate its actual running time through computational experiments. The experimental results indicate that the proposed method is not only polynomial-time solvable in theory but also computationally practical in actual execution time.

%%%%%%%%%%%%%%%%%%%%%%%%%%%%%%%%%%%%%%%%%%%%%%%%%%%%%%%%%%%%%%%%%%%%%%
%% 1. Preliminaries
%%%%%%%%%%%%%%%%%%%%%%%%%%%%%%%%%%%%%%%%%%%%%%%%%%%%%%%%%%%%%%%%%%%%%%
\section{Preliminaries}\label{sec:preliminaries}

We consider rooted unordered trees of dimension $2$. For a set $S$, let $|S|$ denote the number of elements in $S$. Let $\Sigma$ and $\Lambda$ be finite alphabets of vertex labels and edge labels, respectively, and let $X$ be an infinite alphabet of variable labels. We assume that $(\Sigma \cup \Lambda) \cap X = \emptyset$. Let $\ut_{\Sigma,\Lambda}$ denote the set of all rooted unordered trees whose vertices and edges are labeled with symbols in $\Sigma$ and $\Lambda$, respectively.

\begin{definition}[Linear unordered term tree pattern]
Let $T=(V_T,E_T)$ be a rooted unordered tree. Let $E_t$ and $H_t$ be a partition of $E_T$. An unordered term tree pattern is a triple $t=(V_t,E_t,H_t)$ obtained from $T$, where $V_t = V_T$. Elements of $V_t$, $E_t$, and $H_t$ are called vertices, edges, and variables, respectively, and are labeled with symbols in $\Sigma$, $\Lambda$, and $X$, respectively.

The root, leaves, parent-child relation, and height of $t$ are those of the underlying rooted unordered tree $T$. If $(v,v') \in E_t$, then $(v,v')$ denotes the edge from $v$ to $v'$. If $[v,v'] \in H_t$, then $[v,v']$ denotes the variable from $v$ to $v'$, where $v$ is the parent port and $v'$ is the child port. A non-leaf vertex is called an internal vertex. The pattern $t$ is said to be \emph{linear} if no two variables in $H_t$ share the same variable label. In this paper, we deal only with linear unordered term tree patterns.
\end{definition}

\begin{definition}[Height-constrained variables]
Let $X^\h$ be an infinite subset of $X$. For each pair of integers $i,j$ with $1 \le i \le j$, let $X^{\h(i,j)}$ be an infinite subset of $X^\h$. We assume that $X^\h=\bigcup_{1 \le i \le j} X^{\h(i,j)}$ and $X^{\h(i,j)} \cap X^{\h(i',j')}=\emptyset$ holds whenever $(i,j) \neq (i',j')$.

A variable label in $X^{\h(i,j)}$ is called an \emph{$(i,j)$-height-constrained variable label}, or simply an \emph{$(i,j)$-HC-variable label}. A variable $[v,v']$ in an unordered term tree pattern is called an \emph{$(i,j)$-height-constrained variable}, or simply an \emph{$(i,j)$-HC-variable}, if its label belongs to $X^{\h(i,j)}$.

We write $\luttp_{\Sigma,\Lambda,X^\h}$ for the set of all linear unordered term tree patterns in which vertices are labeled with symbols in $\Sigma$, edges are labeled with symbols in $\Lambda$, and variables are labeled with symbols in $X^\h$.
\end{definition}

\begin{definition}[Generalized bindings and substitutions]
Let $x$ be a variable label in $X^{\h(i,j)}$ for some integers $1 \le i \le j$, and let $g$ be an unordered term tree pattern in $\luttp_{\Sigma,\Lambda,X^\h}$ with at least two vertices. Let $r_g$ be the root of $g$, and let $q \in V_g \setminus \{r_g\}$. We say that $(g,q)$ is \emph{admissible} for $x$ if $dist_g(r_g,q) \ge i$ and $height(g) \le j$. If $(g,q)$ is admissible for $x$, then $x := (g,q)$ is called a \emph{generalized binding} for $x$, and $dist_g(r_g,q)$ is called its trunk length.

Let $f=(V_f,E_f,H_f)$ be an unordered term tree pattern in $\luttp_{\Sigma,\Lambda,X^\h}$, and let $h=[v,v']$ be an $(i,j)$-HC-variable of $f$ labeled with $x \in X^{\h(i,j)}$. For a generalized binding $x := (g,q)$, let $g'$ be a copy of $g$, and let $r'$ and $q'$ be the vertices of $g'$ corresponding to $r_g$ and $q$, respectively. The pattern obtained by applying $x := (g,q)$ to $h$ is defined by removing $h$ from $H_f$, identifying $v$ with $r'$, and identifying $v'$ with $q'$. The labels of the identified vertices inherit the labels of $v$ and $v'$ in $f$.

A \emph{generalized substitution} is a finite set $\theta = \{x_1 := (g_1,q_1), \ldots, x_m := (g_m,q_m)\}$ such that the variable labels $x_1,\ldots,x_m$ are pairwise distinct, and each $g_\ell$ contains no variable whose label belongs to $\{x_1,\ldots,x_m\}$. The instance $f\theta$ is the unordered term tree pattern obtained by applying all generalized bindings in $\theta$ simultaneously, and its root is the root of $f$.
\end{definition}

\begin{figure}[t]
\centering
\includegraphics[scale=0.48]{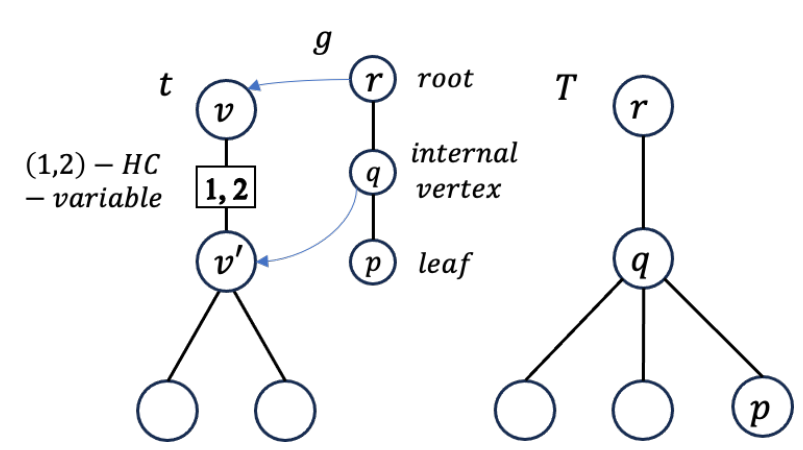}
\caption{An instance matched only under the generalized model. The square denotes an $(1,2)$-HC-variable whose endpoints are the parent port $v$ and the child port $v'$. Since $v'$ is matched to the internal vertex $q$ of the binding tree $g$, the instance is not matched under the previous leaf-restricted model but is matched under the generalized model.}\label{fig:example-match-unmatch}
\end{figure}

As shown in Fig.~\ref{fig:example-match-unmatch}, the child port of an HC-variable may be matched to an internal vertex of a binding tree under the generalized model.
Such a situation naturally occurs when a variable represents a structured component that attaches through an internal connection point rather than a leaf, as in abstract syntax representations or chemical structure trees.
For $t \in \luttp_{\Sigma,\Lambda,X^\h}$, let $L(t)=\{\,T \in \ut_{\Sigma,\Lambda} \mid T \cong t\theta \text{ for some generalized substitution } \theta\,\}$.
We say that an unordered tree $T$ matches an unordered term tree pattern $t$ if $T \in L(t)$.

The membership problem for $\luttp_{\Sigma,\Lambda,X^\h}$ under generalized bindings is defined as follows: given $t \in \luttp_{\Sigma,\Lambda,X^\h}$ and $T \in \ut_{\Sigma,\Lambda}$, decide whether $T \in L(t)$. Unlike the previous rooted model, the child port of an HC-variable may correspond to any non-root vertex of a binding tree, not necessarily a leaf, while the height constraint is still imposed on the whole binding tree. The previous definition is obtained as a special case by restricting the distinguished vertex $q$ to a leaf of the binding tree.

%%%%%%%%%%%%%%%%%%%%%%%%%%%%%%%%%%%%%%%%%%%%%%%%%%%%%%%%%%%%%%%%%%%%%%
%% 3. Polynomial-Time Algorithms
%%%%%%%%%%%%%%%%%%%%%%%%%%%%%%%%%%%%%%%%%%%%%%%%%%%%%%%%%%%%%%%%%%%%%%
\section{A Polynomial-Time Matching Algorithm under Generalized Bindings}\label{sec:algorithm}

Throughout this section, we assume $|\Sigma|=|\Lambda|=1$ for clarity of presentation. The algorithm extends straightforwardly to general alphabets $\Sigma$ and $\Lambda$ by additionally requiring that vertex and edge labels match at each step; the asymptotic complexity is unchanged.

Let $t$ be a rooted unordered term tree pattern in $\luttp_{\Sigma,\Lambda,X^\h}$ and let $T$ be a rooted unordered tree in $\ut_{\Sigma,\Lambda}$. In this section, we present a polynomial-time algorithm for deciding whether $T \in L(t)$ under generalized bindings. Processing $T$ in a bottom-up order, the algorithm computes for each vertex $u$ two kinds of information: $\vs(u)$, which records exact matches of subpatterns of $t$ at $u$, and $\ins(u)$, which stores compressed inheritance information.

The main difficulty in the generalized model is that inherited states must retain not only local feasibility but also minimum height information; in particular, triples of the form $(u',0,j')$ with $j' > 0$ must be propagated upward. To handle this, for each relevant local matching instance we construct a weighted bipartite graph with dummy vertices and solve a bottleneck matching problem. The optimum bottleneck value gives the minimum height that must be inherited upward.

For the time-complexity analysis, we assume that each local bottleneck matching instance is solved by the algorithm of Gabow and Tarjan~\cite{gabow-tarjan1988}, which improves upon Bhat's $O(n^{1/2} m \log n)$-time algorithm~\cite[p.~412]{gabow-tarjan1988} and runs in $O((n \log n)^{1/2}m)$ time. Since both the size and the number of local instances are polynomially bounded by the sizes of $t$ and $T$, the overall algorithm runs in polynomial time.

For $u \in V_t$, let $t[u]$ denote the subpattern of $t$ rooted at $u$, and for $v \in V_T$, let $T[v]$ denote the subtree of $T$ rooted at $v$. If $d$ is a descendant of $v$, then $T[v]-T[d]$ denotes the tree obtained from $T[v]$ by deleting all proper descendants of $d$, and $height_{T[v]}(d)$ denotes its height.

\begin{definition}[Matching rule]
Let $u'$ be an internal vertex of $t$, and let $c'_1,\ldots,c'_m$ be the children of $u'$. For each child $c'_\ell$ $(1 \le \ell \le m)$, set $J(c'_\ell)=c'_\ell$ if $c'_\ell$ is not the child port of an HC-variable. If $c'_\ell$ is the child port of an $(i,j)$-HC-variable, then set $J(c'_\ell)=(c'_\ell,i,j)$.
The \emph{matching rule} associated with $u'$ is the expression $u' \Leftarrow \{J(c'_1),\ldots,J(c'_m)\}$. Whether $u'$ is the child port of an HC-variable is determined from $t$ and is not encoded in the rule itself. The set of all such rules is denoted by $Rule(t)$.
\end{definition}

\begin{definition}[Exact-match information]
For $v \in V_T$, let $\vs(v)=\{\,u' \in V_t \mid t[u'] \text{ matches } T[v]\,\}$.
\end{definition}

\begin{definition}[Full inheritance information]
Let $v \in V_T$. The full inheritance information at $v$, denoted by $\widehat{\ins}(v)$, is the set of all triples $(u',i',j')$ such that $u'$ is the child port of an HC-variable of $t$, and there exists a descendant $d$ of $v$ with $u' \in \vs(d)$, $dist_T(v,d)=i'$, and $height_{T[v]}(d)=j'$.
\end{definition}

A triple $(u',i',j') \in \widehat{\ins}(v)$ means that $u'$ matches at some descendant $d$ of $v$ with $dist_T(v,d)=i'$ and $height_{T[v]}(d)=j'$. For the correctness proof, it is convenient to use $\widehat{\ins}(v)$ first. The actual algorithm uses a compressed form in order to keep the number of states polynomial.

For each child port $u'$ of an HC-variable in $t$, let $tr(u')$ and $ht(u')$ denote the trunk-length lower bound and the height upper bound of the corresponding HC-variable, respectively.

\begin{definition}[Compressed inheritance set]
Let $v \in V_T$. The compressed inheritance set at $v$, denoted by $\ins(v)$, consists of all triples $(u',i'',j^*)$ such that:
\begin{enumerate}
\item $u'$ is the child port of an HC-variable of $t$;
\item $0 \le i'' \le tr(u')-1$;
\item if $i'' < tr(u')-1$, then $(u',i'',j') \in \widehat{\ins}(v)$ for some $j'$, and $j^*$ is the minimum such value;
\item if $i'' = tr(u')-1$, then $(u',i',j') \in \widehat{\ins}(v)$ for some $i' \ge tr(u')-1$, and $j^*$ is the minimum such value;
\item $j^* \le ht(u')-1$.
\end{enumerate}
\end{definition}

Thus, for each child port $u'$, $\ins(v)$ merges all distances at least $tr(u')-1$ into the single value $tr(u')-1$, and for each retained distance stores only the minimum inherited height. Hence $\ins(v)$ removes redundant information from $\widehat{\ins}(v)$ while preserving the information needed for matching.

\begin{lemma}
Let $v$ be a vertex of $T$. For the purpose of deciding membership under generalized bindings, the full inheritance information $\widehat{\ins}(v)$ can be replaced by the compressed inheritance set $\ins(v)$ without loss of necessary matching information.
\end{lemma}

\begin{definition}[Local weighted bipartite graph]\label{def:local-weighted-bipartite-graph}
Let $u'$ be an internal vertex of $t$ with matching rule
$u' \Leftarrow \{J(c'_1),\ldots,J(c'_m)\}$,
and let $v$ be an internal vertex of $T$ with children $d_1,\ldots,d_n$. If $n<m$, then the rule is not applicable at $v$. Assume $n \ge m$, and define
$G_W(u',v)=(X \cup Y,E,w)$
as follows. The left part $X$ consists of $d_1,\ldots,d_n$, and the right part $Y$ consists of $J(c'_1),\ldots,J(c'_m)$ together with $n-m$ dummy vertices.

For $d_a \in X$ and a non-dummy pattern-side item $J(c'_\ell) \in Y$, add an edge as follows:
\begin{enumerate}
\item if $J(c'_\ell)=c'_\ell$, then $(d_a,J(c'_\ell)) \in E$ iff $c'_\ell \in \vs(d_a)$;
\item if $J(c'_\ell)=(c'_\ell,i,j)$, then $(d_a,J(c'_\ell)) \in E$ iff $(c'_\ell,i-1,j') \in \ins(d_a)$ for some $j' \le j-1$.
\end{enumerate}
Each such edge has weight $0$.

For $d_a \in X$ and a dummy vertex $\delta \in Y$, always add $(d_a,\delta)$ to $E$. If there is no $\ell$ such that $J(c'_\ell)=(c'_\ell,i,j)$, then define $w(d_a,\delta)=0$. Otherwise, let
$j_{\max}=\max\{\,j \mid J(c'_\ell)=(c'_\ell,i,j)\text{ for some }\ell\,\}$,
and define
\[
w(d_a,\delta)=
\begin{cases}
0 & \text{if } height(T[d_a])<j_{\max},\\
height(T[d_a]) & \text{otherwise}.
\end{cases}
\]
\end{definition}

\begin{definition}[Bottleneck value of a local matching instance]
Let $G_W(u',v)=(X \cup Y,E,w)$ be the local weighted bipartite graph for $u'$ and $v$. We say that the matching rule of $u'$ is \emph{applicable} at $v$ if $G_W(u',v)$ has a perfect matching. If so, define
$b(u',v)=\min_M \max_{e \in M} w(e)$,
where $M$ ranges over all perfect matchings of $G_W(u',v)$. A perfect matching attaining this value is called a bottleneck matching.
\end{definition}

If the matching rule of $u'$ is applicable at $v$, then $b(u',v)$ is the minimum height that must be inherited upward. Moreover, $u' \in \vs(v)$ holds if and only if $b(u',v)=0$. If $u'$ is the child port of an HC-variable, then $(u',0,0)$ is inserted into $\ins(v)$ when $b(u',v)=0$, and $(u',0,b(u',v))$ is inserted into $\ins(v)$ when $b(u',v)>0$.

The propagation step from the children of a tree vertex to the vertex itself is unchanged from the previous leaf-restricted model. Thus, for each vertex $v$ of $T$, the algorithm first propagates inherited information from the children of $v$, then computes additional triples of the form $(u',0,j')$ generated by applicable matching rules via bottleneck matchings, and finally compresses the result into $\ins(v)$.

For a leaf $v$ of $T$, we initialize $\vs(v)$ to the set of leaves of $t$ that match $v$, and $\ins(v)$ to the set of all triples $(u',0,0)$ such that $u' \in \vs(v)$ and $u'$ is the child port of an HC-variable. In Algorithm~\ref{alg:matching}, the procedures \textsc{Propagate}, \textsc{LB-Update}, and \textsc{Compress} use the already computed sets $\vs(d)$ and $\ins(d)$ for all children $d$ of the current vertex $v$.

\begin{algorithm}[t]
\caption{\textsc{Match-Generalized-Bindings}}
\label{alg:matching}
%\footnotesize
\begin{algorithmic}[1]
\REQUIRE $t \in \luttp_{\Sigma,\Lambda,X^\h}$ and $T \in \ut_{\Sigma,\Lambda}$
\ENSURE ``yes'' iff $T \in L(t)$
\STATE Let $r_t$ and $r_T$ be the roots of $t$ and $T$
\STATE Compute $Rule(t)$
\STATE $U \gets V_T$

\FOR{each vertex $v \in V_T$}
    \IF{$v$ is a leaf of $T$}
        \STATE Initialize $\vs(v)$ and $\ins(v)$
        \STATE $U \gets U \setminus \{v\}$
    \ELSE
        \STATE $\vs(v) \gets \emptyset$; $\ins(v) \gets \emptyset$
    \ENDIF
\ENDFOR

\WHILE{$U \neq \emptyset$}
    \FOR{each vertex $v \in U$ whose children are all in $V_T \setminus U$}
        \STATE $I_{\mathrm{prop}}(v) \gets \textsc{Propagate}(v)$
        \STATE $(V_{\mathrm{new}}(v), I_{\mathrm{new}}(v)) \gets \textsc{LB-Update}(v, Rule(t))$
        \STATE $\vs(v) \gets V_{\mathrm{new}}(v)$
        \STATE $\ins(v) \gets \textsc{Compress}\!\left(I_{\mathrm{prop}}(v) \cup I_{\mathrm{new}}(v)\right)$
        \STATE $U \gets U \setminus \{v\}$
    \ENDFOR
\ENDWHILE

\IF{$r_t \in \vs(r_T)$}
    \RETURN \textbf{yes}
\ELSE
    \RETURN \textbf{no}
\ENDIF
\end{algorithmic}
\end{algorithm}

In Algorithm~\ref{alg:matching}, \textsc{Propagate} performs the same inheritance propagation as in the previous leaf-restricted model, \textsc{LB-Update} constructs weighted bipartite graphs for locally relevant matching rules, computes their bottleneck values, inserts the resulting exact-match information into $\vs(v)$, and generates new triples of the form $(u',0,j')$ for $\ins(v)$, and \textsc{Compress} replaces the resulting inheritance information by its compressed form.

Let $G_W$ be a weighted bipartite graph. We define $BottleneckValue(G_W)$ as follows. If $G_W$ has no perfect matching, then $BottleneckValue(G_W)=\infty$. Otherwise,
\[
BottleneckValue(G_W)=\min_M \max_{e \in M} w(e),
\]
where $M$ ranges over all perfect matchings of $G_W$.
In the theoretical analysis, this value is computed by the algorithm of Gabow and Tarjan~\cite{gabow-tarjan1988}, whereas in the computational experiments it is computed by a straightforward count-up method.

\begin{algorithm}[t]
\caption{\textsc{LB-Update}}
\label{alg:lb-update}
%\footnotesize
\begin{algorithmic}[1]
\REQUIRE A vertex $v$ of $T$, the sets $\vs(d_1)$, $\ins(d_1)$, $\ldots,$ $\vs(d_n)$, $\ins(d_n)$ for all children $d_1,\ldots,d_n$ of $v$, and the rule set $Rule(t)$
\ENSURE The sets $\vs_{\mathrm{new}}(v)$ and $\ins_{\mathrm{new}}(v)$
\STATE $\vs_{\mathrm{new}}(v) \leftarrow \emptyset$
\STATE $\ins_{\mathrm{new}}(v) \leftarrow \emptyset$
\STATE Let $d_1,\ldots,d_n$ be all children of $v$
\FOR{each matching rule $\rho$ in $Rule(t)$}
    \STATE Let $\rho$ be the rule $u' \Leftarrow \{J(c'_1),\ldots,J(c'_m)\}$
    \IF{$n < m$}
        \STATE \textbf{continue}
    \ENDIF
    \STATE Construct the local weighted bipartite graph $G_W(u',v)$ according to Definition~\ref{def:local-weighted-bipartite-graph}
    \STATE $\beta \leftarrow BottleneckValue(G_W(u',v))$
    \IF{$\beta < \infty$}
        \IF{$\beta = 0$}
            \STATE $\vs_{\mathrm{new}}(v) \leftarrow \vs_{\mathrm{new}}(v) \cup \{u'\}$
        \ENDIF
        \IF{$u'$ is the child port of an HC-variable of $t$}
            \STATE $\ins_{\mathrm{new}}(v) \leftarrow \ins_{\mathrm{new}}(v) \cup \{(u',0,\beta)\}$
        \ENDIF
    \ENDIF
\ENDFOR
\RETURN $(\vs_{\mathrm{new}}(v), \ins_{\mathrm{new}}(v))$
\end{algorithmic}
\end{algorithm}

Algorithm~\ref{alg:lb-update} examines each matching rule at the current tree vertex $v$. If the corresponding local weighted bipartite graph has a perfect matching, let $\beta=b(u',v)$. Then $u'$ is inserted into $\vs_{\mathrm{new}}(v)$ if and only if $\beta=0$. If $u'$ is the child port of an HC-variable, then $(u',0,\beta)$ is inserted into $\ins_{\mathrm{new}}(v)$. Thus, \textsc{LB-Update} computes locally generated exact-match information and inheritance triples of the form $(u',0,j')$.

\begin{lemma}
For every internal vertex $v$ of $T$, Algorithm~\ref{alg:lb-update} correctly computes all newly generated exact-match information at $v$ and all newly generated inheritance triples of the form $(u',0,j')$ arising from locally applicable matching rules at $v$.
\end{lemma}

Since inheritance propagation depends only on the distances and heights recorded in $\ins(d_a)$, and not on whether child ports correspond to leaves or internal vertices, \textsc{I\_Set\_Making} of~\cite{matsushita-eskm2025,shoudai-ieice2017} applies without modification under generalized bindings.
Algorithm~\ref{alg:propagate} performs the same inheritance propagation as in the previous model, carrying upward all triples originating below the current vertex while updating their distance and height values.
Algorithm~\ref{alg:compress} removes redundant inheritance information. For each child port $u'$, it keeps only the minimum inherited height for each retained distance, merges all distances at least $tr(u')-1$ into the single value $tr(u')-1$, and discards triples whose inherited height is at least $ht(u')-1$.

\begin{algorithm}[t]
\caption{\textsc{Propagate}}
\label{alg:propagate}
\begin{algorithmic}[1]
\REQUIRE A vertex $v$ of $T$, and the sets $\vs(d_1)$, $\ins(d_1)$,
  $\ldots,$ $\vs(d_n)$, $\ins(d_n)$ for all children $d_1,\ldots,d_n$ of $v$
\ENSURE The propagated inheritance information $\ins_{\mathrm{prop}}(v)$
\STATE $CS(d_a) \leftarrow \vs(d_a) \cup \ins(d_a)$ for each child $d_a$
\STATE $\mathit{CSs} \leftarrow (CS(d_1),\ldots,CS(d_n))$
\STATE $\ins_{\mathrm{prop}}(v) \leftarrow \textsc{I\_Set\_Making}(v,\,\mathit{CSs},\,H_t)$~\cite{matsushita-eskm2025,shoudai-ieice2017}
\RETURN $\ins_{\mathrm{prop}}(v)$
\end{algorithmic}
\end{algorithm}

\begin{algorithm}[t]
\caption{\textsc{Compress}}
\label{alg:compress}
%\footnotesize
\begin{algorithmic}[1]
\REQUIRE A set $S$ of inheritance triples
\ENSURE The compressed inheritance set $\textsc{Compress}(S)$

\STATE $C \leftarrow \emptyset$

\FOR{each child port $u'$ of an HC-variable of $t$}
    \FOR{$i'' \leftarrow 0$ to $tr(u')-2$}
        \STATE $M \leftarrow \{\, j' \mid (u',i'',j') \in S \,\}$
        \IF{$M \neq \emptyset$ and $\min M \le ht(u')-1$}
            \STATE $C \leftarrow C \cup \{(u',i'',\min M)\}$
        \ENDIF
    \ENDFOR

    \STATE $M \leftarrow \{\, j' \mid (u',i',j') \in S \text{ for some } i' \ge tr(u')-1 \,\}$
    \IF{$M \neq \emptyset$ and $\min M \le ht(u')-1$}
        \STATE $C \leftarrow C \cup \{(u',tr(u')-1,\min M)\}$
    \ENDIF
\ENDFOR

\RETURN $C$
\end{algorithmic}
\end{algorithm}

\begin{theorem}
Algorithm~\ref{alg:matching} correctly decides whether $T \in L(t)$ holds under generalized bindings.
\end{theorem}
\begin{proof}
The proof is by bottom-up induction on the vertices of $T$.
For a leaf $v$ of $T$, the initialization of $\vs(v)$ and $\ins(v)$ is immediate from the definitions. Assume inductively that, for every child $d$ of an internal vertex $v$, the sets $\vs(d)$ and $\ins(d)$ have already been computed correctly. Then Algorithm~\ref{alg:propagate} correctly propagates inheritance information from the children of $v$, Algorithm~\ref{alg:lb-update} correctly computes all locally generated exact-match and inheritance information at $v$, and Algorithm~\ref{alg:compress} replaces the resulting inheritance information by an equivalent compressed form. Hence both $\vs(v)$ and $\ins(v)$ are computed correctly.
If $r_t \in \vs(r_T)$, then $T \in L(t)$ by definition of $\vs$. Conversely, if $T \in L(t)$, then by induction on the structure of the corresponding substitution, one can verify that $r_t \in \vs(r_T)$.
\end{proof}

\begin{theorem}
Under the assumption that each local bottleneck matching instance is solved by the algorithm of Gabow and Tarjan~\cite{gabow-tarjan1988}, the membership problem under generalized bindings is solvable in polynomial time.
\end{theorem}

\begin{proof}
The analysis follows the framework of the previous leaf-restricted algorithm~\cite{matsushita-eskm2025}. The propagation and compression steps remain polynomial, since for each vertex $v$ of $T$, the compressed inheritance set $\ins(v)$ has polynomial size in $|t|$.

The only new cost arises in Algorithm~\ref{alg:lb-update}. For each internal vertex $v$ of $T$ and each matching rule of $t$, we solve one local bottleneck matching instance on $G_W(u',v)$. Let $d=\deg_T(v)$. Then $G_W(u',v)$ has $O(d)$ vertices on each side and $O(d^2)$ edges, so by Gabow and Tarjan~\cite{gabow-tarjan1988}, each instance is solved in
$O\!\left((d\log d)^{1/2}\cdot d^2\right) = O\!\left(d^{5/2}\sqrt{\log d}\right)$.
Since $|Rule(t)|=O(|V_t|)$, the total cost over all vertices of $T$ is
\[
O\!\left(|V_t| \sum_{v\in V_T}\deg_T(v)^{5/2}\sqrt{\log \deg_T(v)}\right).
\]

Let $N=|V_T|$ and $\Delta=\max_{v\in V_T}\deg_T(v)$. Since $\deg_T(v)\le \Delta$ and $\sum_{v\in V_T}\deg_T(v)=N-1$, this is bounded by
$O\!\left(|V_t|\,N\,\Delta^{3/2}\sqrt{\log \Delta}\right)$,
and hence also by
$O\!\left(|V_t|\,|V_T|^{5/2}\sqrt{\log |V_T|}\right)$.
Therefore, the membership problem under generalized bindings is solvable in polynomial time.
\end{proof}

%In the next section, we report computational experiments for the proposed algorithm and evaluate its running time.

%%%%%%%%%%%%%%%%%%%%%%%%%%%%%%%%%%%%%%%%%%%%%%%%%%%%%%%%%%%%%%%%%%%%%%
%% 3. Computational Experiments
%%%%%%%%%%%%%%%%%%%%%%%%%%%%%%%%%%%%%%%%%%%%%%%%%%%%%%%%%%%%%%%%%%%%%%
\section{Computational Experiments}\label{sec:experiments}

\begin{figure*}[t]
\centering
\includegraphics[width=.72\textwidth]{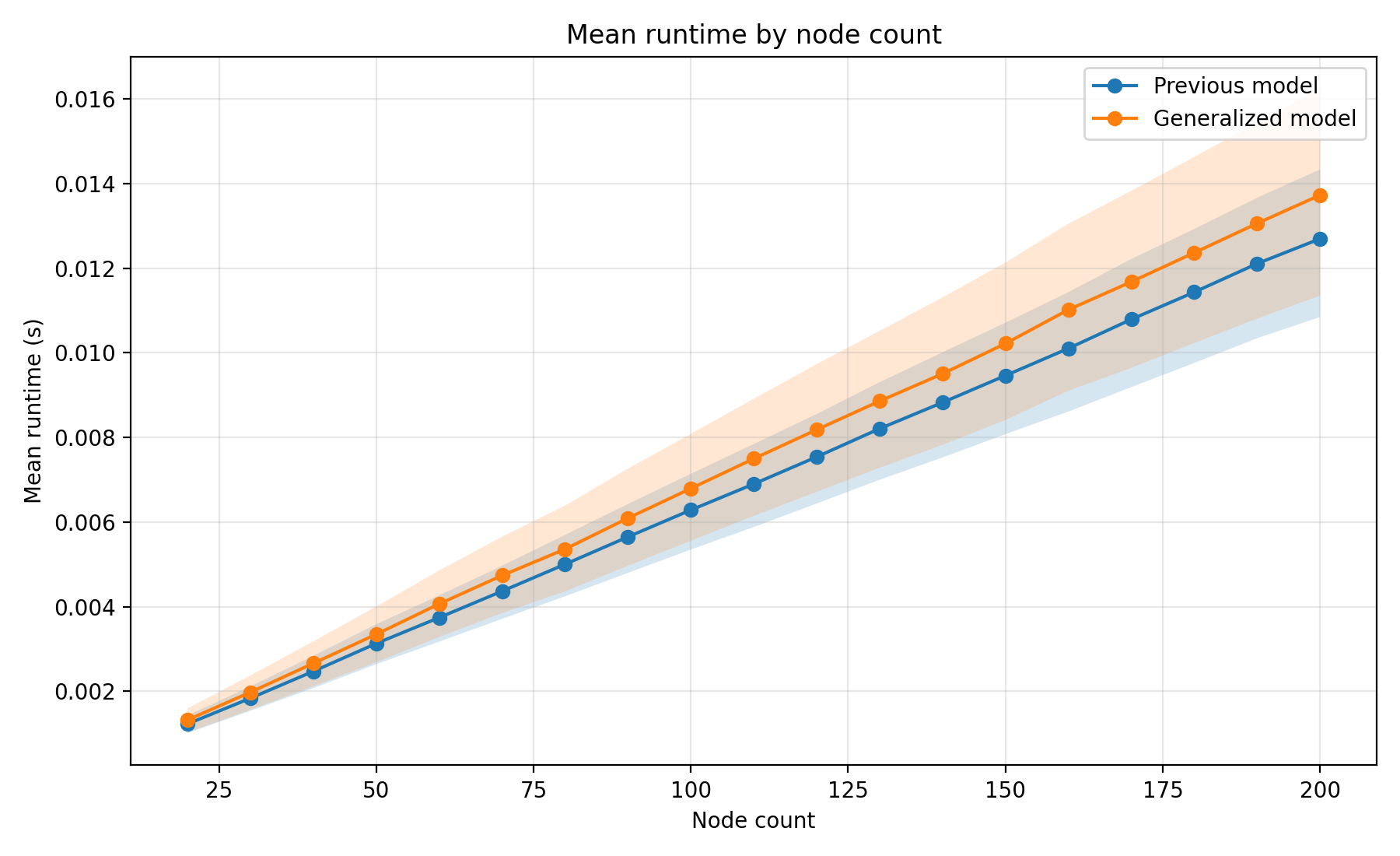}
\caption{Mean running time of the previous model and the generalized model for each node count. Shaded regions indicate the 25th--75th percentile ranges of the running times over the randomly generated instances.}
\label{fig:runtime-mean}
\end{figure*}

\subsection{Implementation Details}

Although the algorithm of Gabow and Tarjan~\cite{gabow-tarjan1988} yields a better theoretical bound for the bottleneck matching subproblem, its implementation is nontrivial in practice. Therefore, in our computational experiments, we adopt a simple count-up method for solving each local bottleneck matching instance. 
More precisely, if no item of the form $J(c'_\ell)=(c'_\ell,i,j)$ occurs, we test only $\beta=0$; otherwise, we enumerate candidate bottleneck values $\beta=0,1,\ldots,J_{\max}-1$ in increasing order, where $J_{\max}$ is the maximum height bound appearing in an HC-variable child over all matching rules of $t$. For each $\beta$, we test whether a perfect matching exists in the subgraph of $G_W(u',v)$ induced by edges of weight at most $\beta$, using the Hopcroft--Karp algorithm~\cite{hopcroft-karp1973}. The smallest feasible value of $\beta$ is then taken as the bottleneck value.

Thus, at most $\max\{1,J_{\max}\}$ perfect-matching tests are performed per instance, each requiring $O(d^{5/2})$ time for a bipartite graph of degree $d$. Hence the implemented algorithm runs in $O\!\left(|V_t|\,N\,\Delta^{3/2}\max\{1,J_{\max}\}\right)$ time, where $N = |V_T|$ and $\Delta = \max_{v\in V_T}\deg_T(v)$. 
When the height bounds of HC-variables are large, this simple implementation requires more perfect-matching tests (proportional to $J_{\max}$), which may increase the overall running time in practice.
Thus, the experimental results show that the proposed framework remains practically workable even when the bottleneck-matching step is implemented in this straightforward manner.

\subsection{Experimental Setup}

To evaluate the practical running time of the proposed algorithm, we conducted computational experiments on randomly generated instances and compared the generalized binding model with the previous leaf-restricted model. For each node count from 20 to 200, we generated 950,000 instances in total and measured the running time of both algorithms under the same implementation environment. For each node count, both target trees and patterns were generated at random, and HC-variables together with their parameters $(i,j)$ were also assigned randomly under the size constraints of the generated patterns. The same family of random instances was used for comparing the previous leaf-restricted model and the generalized model. All experiments were conducted on a machine running Ubuntu 24.04 and equipped with a 13th Gen Intel(R) Core(TM) i9-13900K CPU and 62 GB of RAM. The algorithm was implemented in Python. In the experiments reported below, we focus on the dependence of the running time on the input size.

\subsection{Running-Time Results}

Fig.~\ref{fig:runtime-mean} shows the mean running time of the previous model and the generalized model with respect to the number of nodes. In both models, the running time increases steadily as the input size grows. The generalized model is consistently slower than the previous one, but the difference remains moderate over the entire tested range.
More specifically, the ratio of the mean running time of the generalized model to that of the previous model is approximately 1.08 on average, and stays within a narrow range from 1.0694 to 1.0905 over all tested node counts. The overall mean of the per-node ratios is 1.0803. Thus, although the generalized model requires additional bottleneck-matching computations, the increase in running time is stable and limited.
These results indicate that the proposed generalization preserves practical efficiency while allowing more flexible bindings than the previous leaf-restricted model.
In particular, for the tested range up to 200 nodes, the observed running time remained on the order of $10^{-3}$ to $10^{-2}$ seconds per instance.
The nearly constant runtime ratio suggests that the additional cost introduced by generalized bindings scales proportionally to the baseline matching cost over the tested range.

%%%%%%%%%%%%%%%%%%%%%%%%%%%%%%%%%%%%%%%%%%%%%%%%%%%%%%%%%%%%%%%%%%%%%%
%% 4. Concluding Remarks
%%%%%%%%%%%%%%%%%%%%%%%%%%%%%%%%%%%%%%%%%%%%%%%%%%%%%%%%%%%%%%%%%%%%%%
\section{Concluding Remarks}\label{sec:concl}

In this paper, we presented a polynomial-time algorithm for deciding the membership problem under generalized bindings and reported computational results supporting its practical feasibility. As future work, we plan to extend the framework to richer structural restrictions such as diameter constraints and to study its learnability under computational learning models. This is a natural direction, because generalized bindings increase the expressive power of the pattern language while preserving polynomial-time membership.

\end{document}